\documentclass[submission,copyright,creativecommons]{eptcs}
\usepackage{amsmath}
\usepackage{amssymb}
\usepackage{amsfonts}
\usepackage{amsthm}
\usepackage{enumerate}
\usepackage{amssymb}
\usepackage{textcomp}
\usepackage[english]{babel}
\usepackage[pdftex]{graphicx}
\usepackage{hyperref}
\usepackage{textcomp}
\usepackage[hypcap]{caption}
\usepackage[T1]{fontenc}
\usepackage[latin9]{inputenc}
\usepackage{graphicx}
\usepackage{babel}
\usepackage[usenames,dvipsnames]{color}
\usepackage{mathtools}
\usepackage{braket}
\usepackage{placeins}
\usepackage{color, colortbl}
\usepackage{fancyhdr}
\usepackage{float}
\usepackage{algorithm}
\usepackage{algpseudocode}
\usepackage{stmaryrd}
\usepackage{wrapfig}

\numberwithin{equation}{section}

\theoremstyle{plain}
\newtheorem{theorem}{Theorem}[section]
\newtheorem{lemma}[theorem]{Lemma}
\newtheorem{proposition}[theorem]{Proposition}
\newtheorem{conjecture}[theorem]{Conjecture}

\theoremstyle{definition}
\newtheorem{definition}[theorem]{Definition}
\newtheorem{example}[theorem]{Example}

\theoremstyle{remark}
\newtheorem{remark}[theorem]{Remark}

\date{\today}

\newcolumntype{C}{>{\centering\arraybackslash} m{3cm} }

\usepackage[svgnames]{xcolor}
\usepackage{tikz}
\usepackage{graphics}

\usetikzlibrary{decorations.pathmorphing}
\usetikzlibrary{decorations.markings}
\usetikzlibrary{decorations.pathreplacing}
\usetikzlibrary{arrows}
\usetikzlibrary{shapes}

\pgfdeclarelayer{edgelayer}
\pgfdeclarelayer{nodelayer}
\pgfsetlayers{edgelayer,nodelayer,main}

\tikzstyle{braceedge}=[decorate,decoration={brace,amplitude=10pt}]
\tikzstyle{square box}=[rectangle,fill=white,draw=black,minimum height=6mm,minimum width=6mm,yshift=0.7mm]
\tikzstyle{wire label}=[font=\footnotesize, auto,swap]

\tikzstyle{none}=[inner sep=0pt]
\tikzstyle{gn}=[circle,fill=Lime,draw=Black,line width=0.8 pt]
\tikzstyle{rn}=[circle,fill=Red,draw=Black, line width=0.8 pt]
\tikzstyle{H}=[rectangle,fill=Yellow,draw=Black]
\tikzstyle{line}=[scalar,fill=White,draw=Black]
\tikzstyle{io}=[rectangle,fill=White,draw=Black]
\tikzstyle{block}=[rectangle,fill=Orange,draw=Black]
\tikzstyle{graph}=[circle,fill=White,draw=Black]
\tikzstyle{empty}=[rectangle,fill=White,draw=White]
\tikzstyle{box}=[rectangle,fill=White,draw=Black]
\tikzstyle{dot}=[circle,fill=Black,draw=Black,inner sep=0pt,minimum size=1pt]
\tikzstyle{Dot}=[circle,fill=Black,draw=Black,inner sep=0pt,minimum size=3pt]
\tikzstyle{diam}=[rectangle,fill=Black,draw,yscale=1.2,rotate=45]
\tikzstyle{gangle}=[rectangle,fill=Lime,draw=Black]
\tikzstyle{rangle}=[rectangle,fill=Red,draw=Black]
\tikzstyle{circ}=[circle,fill=none,draw=Black,scale=1.3]
\tikzstyle{ellip}=[ellipse,fill=none,draw=Black,scale=1.3,minimum width =1.3cm]
\tikzstyle{bbox}=[rectangle,fill=Blue,draw=Blue,scale=0.6]
\tikzstyle{gg}=[shape=rectangle,fill=White,draw=Black,dashed]

\tikzstyle{nodev}=[circle,fill=none,draw=Black,scale=1]
\tikzstyle{wirev}=[circle,fill=Black,draw=Black,inner sep=0pt,minimum size=3pt]
\tikzstyle{wirevred}=[circle,fill=Red,draw=Black,inner sep=0pt,minimum size=3pt]

\tikzstyle{simple}=[-,draw=Black]
\tikzstyle{directed}=[
  decoration={markings,mark=at position 1 with {\arrow[scale=2]{>}}},postaction={decorate},shorten >=0.4pt]
\tikzstyle{bdirected}=[
  decoration={markings,
  mark=at position 0.01 with {\arrow[scale=-2]{>}},
  mark=at position 0.99 with {\arrow[scale=2]{>}}
  },postaction={decorate},shorten >=0.4pt]
\tikzstyle{bothdirs}=[bdirected,draw=Black]
\tikzstyle{bothdirsred}=[bdirected,draw=Red]
\tikzstyle{blue}=[-,draw=Blue]
\tikzstyle{redd}=[directed,draw=Red]
\tikzstyle{blued}=[directed,draw=Blue]
\tikzstyle{dash}=[dashed,draw=Black]

\tikzstyle{dotpic}=[scale=0.5]

\tikzstyle{every picture}=[baseline=-0.25em]

\newcommand{\stikz}[2][1]{\scalebox{#1}{
\InputIfFileExists{#2}{}{\input{./tikz/#2}}
}}
\newcommand{\cstikz}[2][1]{\begin{center}\stikz[#1]{#2}\end{center}}

\usepackage{breakurl}
\providecommand{\event}{GaM 2015} % Name of the event you are submitting to

\newcommand{\presec}{\vspace{-3.5mm}}
\newcommand{\postsec}{\vspace{-3.5mm}}

\title{!-Graphs with Trivial Overlap are Context-Free}
\author{Aleks Kissinger
\institute{Department of Computer Science\\
University of Oxford\\
Oxford, United Kingdom}
\email{aleks.kissinger@cs.ox.ac.uk}
\and
Vladimir Zamdzhiev
\institute{Department of Computer Science\\
University of Oxford\\
Oxford, United Kingdom}
\email{\quad vladimir.zamdzhiev@cs.ox.ac.uk}
}
\def\titlerunning{!-Graphs with Trivial Overlap are Context-Free}
\def\authorrunning{A. Kissinger \& V. Zamdzhiev}
\begin{document}
\maketitle

\begin{abstract}
String diagrams are a powerful tool for reasoning about composite structures in symmetric monoidal categories. By representing string diagrams as graphs, equational reasoning can be done automatically by double-pushout rewriting. !-graphs give us the means of expressing and proving properties about whole families of these graphs simultaneously. While !-graphs provide elegant proofs of surprisingly powerful theorems, little is known about the formal properties of the graph languages they define. This paper takes the first step in characterising these languages by showing that an important subclass of !-graphs---those whose repeated structures only overlap trivially---can be encoded using a (context-free) vertex replacement grammar.
\end{abstract}

% We conjecture that context-freeness actually characterises trivial-overlap for !-graphs, and comment on the possibility of extending the reasoning techniques employed for !-graphs to all context-free string graph languages.

  \presec
  \section{Introduction}\label{sec:intro}
  \postsec
  % -*- root: main.tex -*-

\begin{figure}
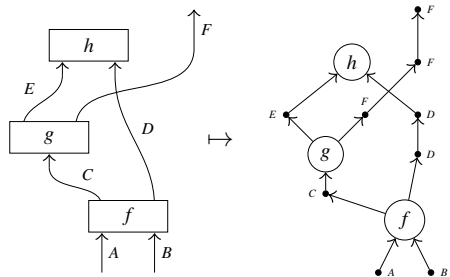
%{r}{0.4\textwidth}
  \centering
  \stikz[0.7]{string_diagram_to_graph.tikz}
  \caption{A string diagram, and its encoding as a string graph}\label{fig:string-diagram-encode}
\end{figure}
String diagrams are essentially directed graphs, but instead of edges they have a more flexible notion of \textit{wires}. Wires can be left open at one or both ends to form inputs and outputs to the diagram, or can be connected to themselves to form a circle. They have a formal semantics given by monoidal category theory~\cite{joyal_street}, and have found applications in many fields. For example: in models of concurrency, they give an elegant presentation of Petri nets with boundary~\cite{Sobocinski:2010aa}, in computational linguistics, they are used to compute compositional semantics for sentences~\cite{LambekvsLambek}, in control theory, they represent signal-flow diagrams~\cite{Baez2014a,Bonchi2015}, and in theoretical physics, they provide the formal language of categorical quantum mechanics~\cite{picturalism}. A \textit{string graph} encodes a string diagram as a typed, directed graph, by replacing the wires with chains of edges containing special dummy vertices called \textit{wire-vertices} (see Fig.~\ref{fig:string-diagram-encode}).
By contrast, the ``real'' vertices, labelled here $f, g$ and $h$ are called \textit{node-vertices}. String graphs---originally introduced under the name ``open graphs'' in~\cite{open_graphs1}---have the advantage that they are purely combinatoric (as opposed to geometric) objects. %, and can be transformed using double-pushout (DPO) rewriting.
Equational reasoning on string diagrams is done by replacing sub-diagrams. The presence of dummy vertices in string graphs allows this to be done with double-pushout (DPO) graph rewriting. For example, in:
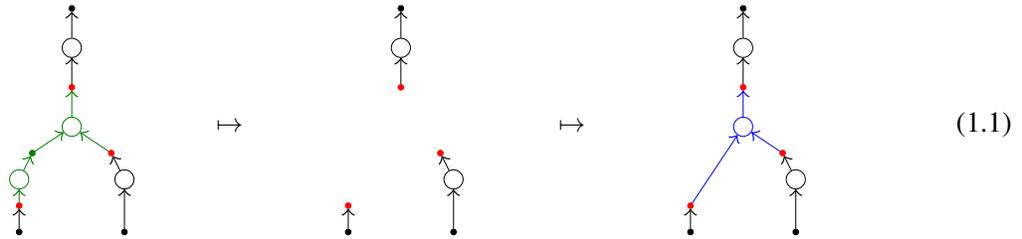
\begin{equation}\label{eq:subst-example}
  \stikz[0.7]{sg-sub.tikz}
\end{equation}
the LHS (shown in green above) is replaced by the RHS (blue), using the common boundary (red). This type of rewriting for string graphs is implemented in the graphical proof assistant Quantomatic~\cite{quanto-cade}.

Often, one wants to reason not just about single string graphs, but entire families of them. However, informal notions of graphs with repetition are ill-suited for automated tools. To address this problem, \textit{!-graphs} (pronounced ``bang-graphs'') were introduced in~\cite{DD1} and formalised for string graphs in~\cite{pattern_graphs}. The idea behind !-graphs is that certain marked subgraphs (along with their adjacent edges) can be repeated any number of times, in a manner somewhat analogous to the Kleene star. These marked subgraphs are called \textit{!-boxes}. For example:
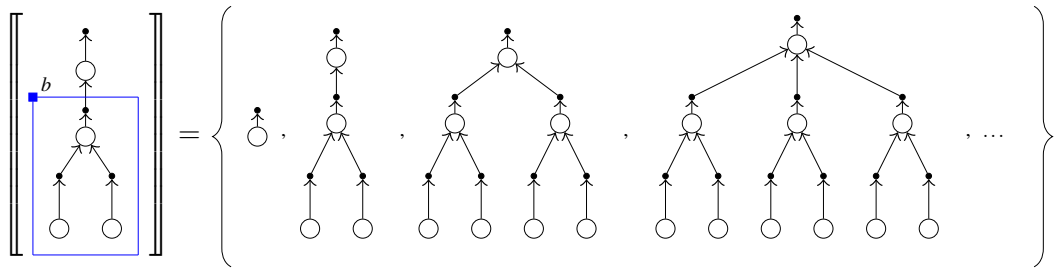
\[  \left \llbracket \stikz[0.7]{bang_graph2.tikz} \right \rrbracket
  =
  \stikz[0.7]{bang_graph_semantics.tikz}
\]
The !-graph in the semantic brackets represents the depicted set of string
graphs.
More precisely, an instance of a !-graph is obtained by repeatedly applying of the two operations EXPAND and KILL on the !-graph:
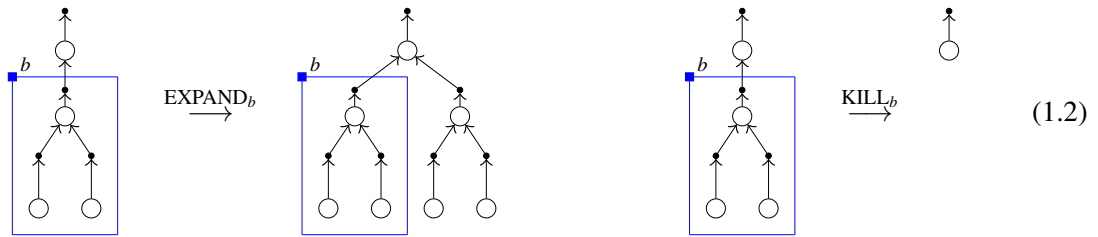
\begin{equation}\label{eq:bbox-ops}
  \stikz[0.7]{bang_graph2.tikz}
  \quad \overset{\textrm{EXPAND}_b}{\longrightarrow} \quad
  \stikz[0.7]{bang_graph2-exp.tikz}
  \qquad\qquad\qquad
  \stikz[0.7]{bang_graph2.tikz}
  \quad \overset{\textrm{KILL}_b}{\longrightarrow} \quad
  \stikz[0.7]{bang_graph2-kill.tikz}
\end{equation}
until all !-boxes have been eliminated. The set of concrete graphs (i.e. those not containing any !-boxes) obtainable from a !-graph via these operations is called the \textit{language} of the !-graph.

We call the set of languages expressible by !-graphs \textbf{BG}. In this paper, we will compare the expressiveness of this language to context-free vertex replacement grammars. We will focus on confluent neighbourhood-controlled embedding grammars, with directions and edge-labels, i.e.~\textbf{C-edNCE} grammars~\cite{graph_grammar_handbook}.

In arbitrary !-graphs, the !-boxes are allowed to nest inside of each other, but also overlap. The latter can generate patently non-context-free behaviour. Thus, it is natural to consider a restricted set of languages, \textbf{BGTO} of \textit{!-graphs with trivial overlap}. In this paper, we will give an algorithm for producing a C-edNCE grammar that is furthermore \textit{linear} (LIN-edNCE) which reproduces the language of any !-graph with trivial overlap. We therefore show the relationship between graph languages illustrated in Fig.~\ref{fig:venn1}.

\begin{figure}
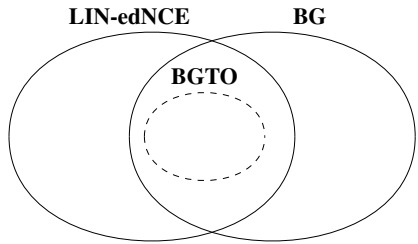
 %{r}{0.4\textwidth}
\centering
\stikz[0.8]{venn_diagram.tikz}
\caption{Relationship between \textbf{LIN-edNCE}, \textbf{BG} and
\textbf{BGTO}} \label{fig:venn1}
\end{figure}

After introducing some preliminaries on string graphs, !-graphs, and C-edNCE grammars in Section~\ref{sec:bg}, we will give an encoding of !-graphs into LIN-edNCE grammars in Section~\ref{sec:encoding}. In that section, we also demonstrate that \textbf{LIN-edNCE} contains languages not in \textbf{BG} (and hence is strictly larger than \textbf{BGTO}). Finally, we give one conjecture and discuss the prospects for reasoning with context-free string graph grammars in a manner analogous to the !-graph case in Section~\ref{sec:fw}.

  \presec
  \section{Preliminaries}\label{sec:bg}
  \postsec

\begin{definition}[Graph \cite{graph_grammar_handbook}]
  A \textit{graph} over an alphabet of node labels $\Sigma$ and an alphabet of edge
  labels $\Gamma$ is a tuple $H = (V, E, \lambda)$, where $V$ is a finite set
  of nodes, $E \subseteq \{(v, \gamma, w) | v, w \in V, v \not= w, \gamma \in
  \Gamma\}$ is the set of edges and $\lambda : V \to \Sigma$ is the node
  labelling function. The set of all graphs with labels $\Sigma, \Gamma$ is denoted $GR_{\Sigma,\Gamma}$.
\end{definition}

\begin{remark}
Note, that this definition of graphs does not allow for self-loops. This is a standard (and convenient) assumption in the C-edNCE literature. By using this notion of graphs, we do not lose any expressivity for string graphs, because any string graph with a self-loop is wire-homeomorphic (see Definition~\ref{def:wire-homeo}) to a string graph with no self-loops.
\end{remark}

\begin{definition}[String Graph] \label{def:string-graph}
  A \textit{string graph} is a graph labelled by the set $\{ N, W \}$, where vertices labelled $N$ are called \textit{node-vertices} and vertices labelled $W$ are called \textit{wire-vertices}, and the following conditions hold:
  (1) there are no edges directly connecting two node-vertices,
  (2) the in-degree of every wire-vertex is at most one and
  (3) the out-degree of every wire-vertex is at most one.
\end{definition}

We will depict wire vertices as small black dots and node-vertices as larger white circles, as we have already done in Section~\ref{sec:intro}.
  For simplicity, we assume there is only one type of node-vertex, but we could easily capture string graphs like~\ref{fig:string-diagram-encode} by taking the labels to be $\{ N_f, N_g, N_h, W \}$, for example.

We also identify those wire-vertices which form the boundary of a string graph:

\begin{definition}[Inputs, Outputs and Boundary Vertices]
A wire-vertex of a string graph $G$ is called an \textit{input} if it has no incoming
edges. A wire-vertex with no outging edges is called an \textit{output}. We denote
with $\textit{In}(G)$ and $\textit{Out}(G)$ the string graphs with no edges whose
vertices are respectively the inputs and outputs of $G$. The boundary of
$G$ is $\textit{Bound}(G) := \textit{In}(G) \cup \textit{Out}(G)$.
\end{definition}

Just as node-vertices represent the boxes in a string diagram, chains of wire-vertices represent the wires. 

\begin{definition}[Wires]
  A connected chain of vertices where each endpoint is either a boundary or a node-vertex, and all other vertices are wire-vertices is called a \textit{closed wire}. A closed wire minus its endpoints is called the \textit{interior} of a closed wire. A cycle consisting only of wire-vertices is called a \textit{circle}.
\end{definition}

Often, it is more natural to consider an equivalence class of string graphs, up to \textit{wire homeomorphism}.

% Next, we provide some definitions which make the notion of wire-homeomorphism precise.

% \begin{definition}[Simple Chains and Cycles]
%   A \textit{simple chain} is a connected, acyclic graph where each vertex has
%   at most one in-edge and one out-edge. A vertex in a chain with only an
%   in-edge or only an out-edge is called an $endpoint$. A \textit{simple cycle}
%   is a connected graph where each vertex has exactly one in-edge and one
%   out-edge.
% \end{definition}

% \begin{definition}[Closed Wire]
%   For a string graph $H$, a \textit{closed wire} $W \subseteq G$ is a simple
%   cycle of wire vertices or a simple chain such that:
%   \begin{enumerate}
%     \item the endpoints of $W$ are either node-vertices or in the boundary of
%       $G$ and
%     \item all other vertices in $W$ are wire-vertices.
%   \end{enumerate}
% \end{definition}

% \begin{definition}[Homeomorphpic wires]
%   Two wires $W$ and $W'$ are said to be \textit{homeomorphic} if they are both circles
% \end{definition}

\begin{definition}[Wire-homeomorphic string graphs] \label{def:wire-homeo}
  Two string graphs $G$ and $G'$ are called \textit{wire-homeo\-morphic}, written $G \sim G'$ if $G'$ can be obtained from $G$ by either merging two adjacent wire-vertices (left) or by splitting a wire-vertex into two adjacent wire-vertices (right) any number of times:
    \[ \stikz[0.7]{two-wires.tikz} \ \ \mapsto \ \ \stikz[0.7]{one-wire.tikz}
      \quad\quad\quad\quad
    \stikz[0.7]{one-wire.tikz} \ \ \mapsto \ \ \stikz[0.7]{two-wires.tikz} \]
\end{definition}

\begin{wrapfigure}{R}{0.4\textwidth}
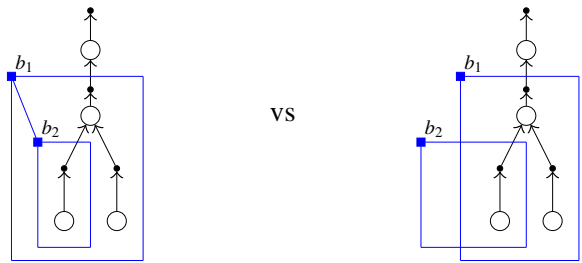

  \centering
  \stikz[0.6]{wire-homeo.tikz}
  \caption{Wire homeomorphism}\label{fig:homeo}
\end{wrapfigure}
That is, two graphs are wire-homeomorphic if the only difference between them is the number of wire-vertices used to represent a wire in the diagram
(see Fig.\ref{fig:homeo}).

% !-graphs (pronounced \textit{bang graphs}) were introduced in
% \cite{pattern_graphs} under the name \textit{pattern graphs} and were later
% renamed to their current name in \cite{merry_dphil} which contains the most
% complete and detailed description of !-graphs.

% A !-graph is a generalised string graph which allows us to represent infinite
% families of string graphs in a formal way. In addition to wire and box
% vertices, !-graphs also have a third type of vertices, called !-vertices. The
% edges coming out of a !-vertex determine a !-box whose contents can be copied
% infinitely many times (while preserving connection relations), thus allowing us
% to represent an infinite set of string graphs in a finite way.

% The motivation behind !-graphs is to remove some of the informalities in
% expressing infinite families of rewrite rules and infinite families of diagrams
% using the familiar "dot dot dot" $(\cdots)$ notation in order to ease the
% development of software proof-assistants.

Next, we provide definitions for !-graphs and related notions. First, we need to identify which subgraphs are legal to put in !-boxes.

\begin{definition}[Open Subgraph]
  A subgraph $O$ of a string graph $H$ is said to be \textit{open} if it is a full subgraph and furthermore $\textit{In}(H \backslash
  O) \subseteq \textit{In}(H)$ and $\textit{Out}(H \backslash O) \subseteq \textit{Out}(H)$.
\end{definition}

\begin{definition}[!-graph]
  A \textit{!-graph} $H$ is a string graph, a partially ordered set $!(H)$ of \textit{!-boxes}, and an open subgraph $B(b) \subseteq H$ for every $b \in !(H)$ such that $b \leq b' \implies B(b) \subseteq B(b')$.
\end{definition}

The notion of openness ensures that we never put only part of a wire inside a !-box, which would cause wires to ``split'' in the middle as we expand the !-box, which violates the arity conditions for wire-vertices from Definition~\ref{def:string-graph}.

% \cstikz[0.7]{split-ends.tikz}

If $b_2 \leq b_1$, this means that $b_2$ is \textit{nested} inside of $b_1$, which we indicate by drawing a line connecting their corners. This is to distinguish from the case where $b_1$ and $b_2$ merely \textit{overlap}.
\[
\stikz[0.7]{bang_graph-nest.tikz}
\qquad\qquad\textrm{vs}\qquad\qquad
\stikz[0.7]{bang_graph-overlap.tikz}
\]

\begin{definition}[Concrete Graph]
A !-graph with no !-vertices is called a \textit{concrete
graph} or simply a string graph.
\end{definition}

% \begin{example}
% We draw !-vertices as blue box-shaped nodes and all their outgoing edges are
% colored in blue. Wire and box vertices are drawn as before (as a small black
% dot and white circle respectively).
% \cstikz[0.7]{bang_graph1.tikz}
% The above depiction of a !-graph is a little cluttered, so instead of drawing
% the edges associated with a !-vertex, we will simply draw a blue box around
% all the nodes which are adjacent to it (hence the name !-box). Like so:
% \cstikz[0.7]{bang_graph2.tikz}
% \end{example}

A !-graph represents an infinite set of concrete string graphs, each of which is obtained after a finite number of applications of the EXPAND and KILL operations presented below.

\begin{definition}[!-box Operations]
The two primitive !-box operations are as follows:
% For a !-graph $G$ and a !-vertex $b \in G$, the three !-box operations are
% functions which produce a new !-graph as follows:
\begin{itemize}
  \item $\textrm{KILL}_b(G) = G' \backslash B(b)$
  \item $\textrm{EXPAND}_b(G) = G \ \sqcup_{KILL_b(G)}\  G'$
\end{itemize}
where $G'$ has the same underlying string graph as $G$, but with $b$ and all of its children removed from $!(G)$. The expression for EXPAND is the disjoint union of $G$ and $G'$, with the vertices, edges, and !-boxes in the common subgraph $\textrm{KILL}_b(G)$ identified.
\end{definition}

As seen in~\eqref{eq:bbox-ops}, $\textrm{KILL}_b$ removes a !-box $b$ and its contents, whereas $\textrm{EXPAND}_b$ inserts a new copy of the contents of $b$. See~\cite{pattern_graphs} for a more detailed description of these operations.

% where $succ(b)$ is the set of all successors of the !-vertex $b$ (which
% includes itself) and the copy operation is the quotient of the disjoint union
% of $G$ with itself over the KILL operation applied to the same !-box.

% In other words, applying a DROP operation removes a !-box, but not its
% contents:
% \cstikz[0.7]{drop.tikz}
% applying a KILL operation removes a !-box and its contents:
% \cstikz[0.7]{kill.tikz}
% applying a COPY operation creates another copy of the !-box and its contents
% which is connected in the same way to the rest of the graph:
% \cstikz[0.7]{tikz/copy.tikz}

% Semantically, a !-graph represents the infinite set of concrete string graphs
% obtained after applying all possible sequences of !-box operations.

\begin{definition}[!-graph Language]
  The \textit{language} of a !-graph $G$ is the set of all concrete string graphs obtained by applying a sequence of !-box operations on $G$.
\end{definition}

% \begin{example}
% The language of the running example is given by:

% \cstikz[0.7]{tikz/bang_graph_semantics.tikz}
% \end{example}

Note that overlapping !-boxes are not explicitly ruled out. However, they can create some odd non-local behaviour. So, na\"ively, one might want to rule out overlap entirely, but it turns out that certain instances of overlap are more innocent than others.

% These we call \textit{trivially overlapping}.

% Finally, we provide definitions which fully characterize the possible
% relationships between a pair of !-boxes $b_1$ and $b_2$.

\begin{definition}[Overlap and Trivial Overlap]
  Given a pair of non-nested !-boxes $b_1$ and $b_2$, we say that $b_1$ and $b_2$ are \textit{overlapping} if $B(b_1) \cap B(b_2) \not = \emptyset$.
  %and otherwise we say that they are \textit{disjoint}.
  $b_1$ and $b_2$ \textit{overlap trivially} if $B(b_1) \cap B(b_2)$ consists of only the interior of zero or more closed wires, where one endpoint is a node-vertex only in $B(b_1)$ and the other is a node-vertex only in $B(b_2)$.
\end{definition}

\begin{wrapfigure}{r}{0.4\textwidth}
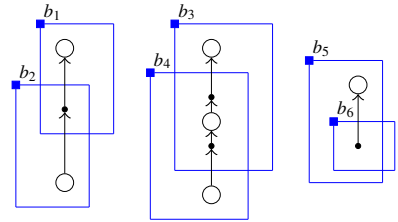

\centering
\stikz[0.65]{bang_relationship.tikz}
\caption{Trivial vs non-trivial overlap}\label{fig:bb_overlap}
\end{wrapfigure}
In particular, $B(b_1) \cap B(b_2)$ can be empty, so trivial overlap generalises no-overlap. In Fig.~\ref{fig:bb_overlap} $b_1$ and $b_2$ overlap trivially, since they only overlap on the interior of a wire connecting node-vertices in the two boxes. $b_3$ and $b_4$ overlap non-trivially, since they overlap on a node-vertex. The most subtle case is $b_5$ and $b_6$, which overlap non-trivially because the shared wire-vertices are not part of a wire whose endpoints are in distinct !-boxes. Of course, all other pairs of !-boxes overlap trivially, as their intersection is just the empty set.

\begin{definition}[\textbf{BGTO}]
  A !-graph where any two non-nested !-boxes overlap trivially is called a \textit{!-graph with trivial overlap}. The set of all languages induced by these !-graphs is called \textbf{BGTO}.
\end{definition}

%\subsection{Vertex replacement grammars}
Next, we introduce C-edNCE graph grammars. Graph grammars are a generalization of context-free grammars for strings. A graph grammar consists of a finite collection of productions which specify instructions on how to generate a family of graphs. Like their context-free string counterparts, context-free graph grammars have better structural, decidability and complexity properties compared to other more expressive mechanisms for graph transformation. We use the same definitions and conventions for graph grammars as presented in \cite{graph_grammar_handbook}, which also describes C-edNCE grammars and their properties in much greater detail.

% \begin{definition}[Graph Language \cite{graph_grammar_handbook}]
%   The set of all graphs over $\Sigma$ and $\Gamma$ is denoted by $GR_{\Sigma,
%   \Gamma}$. The set of all graphs modulo graph isomorphism is denoted by
%   $[GR_{\Sigma, \Gamma}]$. A \textit{graph language} is a subset of
%   $[GR_{\Sigma, \Gamma}]$.
% \end{definition}

C-edNCE grammars are built using graphs with embedding. These provide the information needed to replace a non-terminal node with a new graph, and update connections accordingly.

\begin{definition}[Graph with embedding \cite{graph_grammar_handbook}]
  A \textit{Graph with embedding} over labels $\Sigma, \Gamma$ is a pair $(H,
  C),$ where $H$ is a graph over $\Sigma, \Gamma$ and $C \subseteq \Sigma
  \times \Gamma \times \Gamma \times V_H \times \{in, out\}$. $C$ is called a \textit{connection relation} and its elements are called \textit{connection instructions}. The set of all graphs with embedding over $\Sigma, \Gamma$ is denoted by $GRE_{\Sigma, \Gamma}$.
\end{definition}

Graph grammars operate by substituting a graph (with embedding) for a non-terminal node of another graph. Connection instructions are used to introduce edges connected to the new graph based on edges connected to the non-terminal. A connection instruction as $(\sigma, \beta / \gamma, x, d)$ says to add an edge labelled $\gamma$ connected to the vertex $x$ in the new graph, for every $\beta$-labelled edge connecting a $\sigma$-labelled vertex to the non-terminal. $d$ then indicates whether this rule applies to in-edges or out-edges of the non-terminal. More formally:

\begin{definition}[Graph Substitution \cite{graph_grammar_handbook}]
  Let $(H, C_H), (D, C_D) \in GRE_{\Sigma,\Gamma}$ be two graphs with
  embedding, where $H$ and $D$ are disjoint. Let $v \in V_H$ be a node of $H$.
  The \textit{substitution} of $(D, C_D)$ for $v$ in $(H, C_H)$ is denoted by
  $(H, C_H)[v/(D,C_D)]$ and is given by the graph with embedding whose
  components are:
  \begin{align*}
    V \quad=\quad &(V_H - \{v\} )\cup V_D\\
    E \quad=\quad &\{(x,\gamma,y)\in E_H | x \not = v, y\not = v\} \cup E_D\\
    &\cup \{(w,\gamma,x)\ |\ \exists \beta \in \Gamma : (w,\beta,v) \in E_H
      ,(\lambda_H(w), \beta / \gamma, x,in)\in C_D \}\\
    &\cup \{(x,\gamma,w)\ |\ \exists \beta \in \Gamma : (v,\beta,w) \in E_H
      ,(\lambda_H(w), \beta / \gamma, x,out)\in C_D \}\\
    \lambda(x) \quad=\quad
      &\begin{cases}
	\lambda_H(x) & \text{if } x \in (V_H - \{v\})\\
	\lambda_D(x) & \text{if } x \in V_D
      \end{cases}\\
    C \quad=\quad &\{(\sigma,\beta / \gamma,x,d) \in C_H\ |\ x \not = v\}\\
    &\cup \{(\sigma,\beta / \delta,x,d)\ |\ \exists \gamma \in \Gamma :
    (\sigma,\beta / \gamma,v,d) \in C_H,(\sigma, \gamma / \delta,x,d)\in C_D\}
  \end{align*}
\end{definition}

Thankfully, graphs with embedding allow for a simple graphical presentation (see Fig.~\ref{fig:subst}). We draw graphs as in the previous sections, but with the additional convention that nodes with non-terminal labels are drawn as boxes instead of circles. A single connection instruction is drawn as a pair of edges crossing the outer frame. The edge and node types outside of the frame specify neighbouring node(s) which are connected to the non-terminal in the host graph and the edge(s) inside the frame indicate what new edges will be constructed to and from those nodes.

% We can think of a graph with embedding as just a normal graph with some
% additional information which describes how the graph should be embedded into
% another graph while performing the productions of a graph grammar. In
% particular, a connection instruction $(\sigma, \beta/ \gamma,x,in)$ for a graph
% with embedding $(H, C)$ means that for every $\sigma$ labelled node $w$ in the
% mother graph for which there is a $\beta$ labeled edge going \textbf{in}to the
% non-terminal node $v$ of the mother graph which is being replaced, then the
% embedding process will establish a $\gamma$ labeled edge from $w$ to $x$. This
% should become more clear after referring to fig. \ref{fig:subst} which provides
% an example of graph substitution. The meaning for $(\sigma, \beta/
% \gamma,x,out)$ is analogous.

% If a graph with embedding has no connection instructions (i.e. $C=\emptyset$), then we can think of it as just an ordinary graph.

% Fig.~\ref{fig:subst} shows an example of graph substitution,
% where the newly established edges are colored in red.
  \begin{figure}
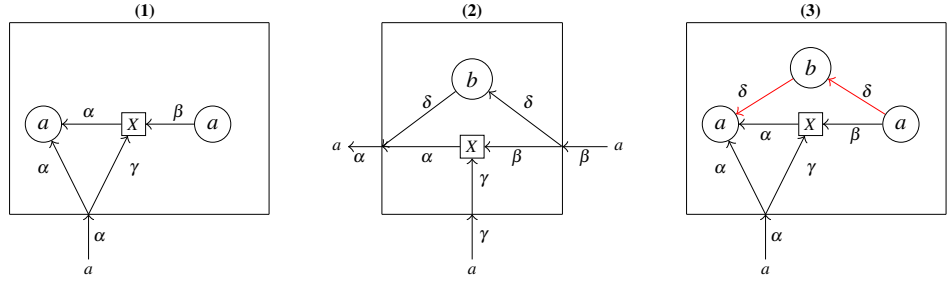

    \cstikz[0.6]{tikz/subst_example1.tikz}
    \caption{
      \textbf{(1)} Graph with embedding $(H,C_H)$
      \textbf{(2)} Graph with embedding $(D,C_D)$
      \textbf{(3)} Result of substitution $(H,C_H)[v/(D,C_D)]$ where $v$ is
        the vertex with non-terminal label in $H$.}\label{fig:subst}
  \end{figure}

Next, we define the concept of an edNCE Graph Grammar. edNCE is an
abbreviation for \textbf{N}eighbourhood \textbf{C}ontrolled \textbf{E}mbedding
for \textbf{d}irected graphs with dynamic \textbf{e}dge relabeling. 

\begin{definition}[edNCE Graph Grammar \cite{graph_grammar_handbook}]
  An \textit{edNCE Graph Grammar} is given by a tuple $G = (\Sigma, \Delta, \Gamma,
  \Omega, P, S)$, where 
    $\Sigma$ is the alphabet of node labels,
    $\Delta \subseteq \Sigma$ is the alphabet of terminal node labels,
    $\Gamma$ is the alphabet of edge labels,
    $\Omega \subseteq \Gamma$ is the alphabet of final edge labels,
    $P$ is a finite set of productions and
    $S \in \Sigma - \Delta$ is the initial nonterminal.
  Productions are of the form $X \rightarrow (D, C)$, where $X \in \Sigma -
  \Delta$ is a non-terminal node and $(D, C) \in GRE_{\Sigma, \Gamma}$ is a
  graph with embedding.
\end{definition}
  We can graphically depict an edNCE grammar by depicting all of its
  productions. A production is drawn exactly like a graph with embedding
  with the addition that we depict its associated non-terminal in the top-left
  part of the frame.
\begin{example}\label{ex:line_graph}
  The following grammar generates the set of all chains of node vertices with an input and no outputs:
  \cstikz[0.7]{unbounded.tikz}
\end{example}

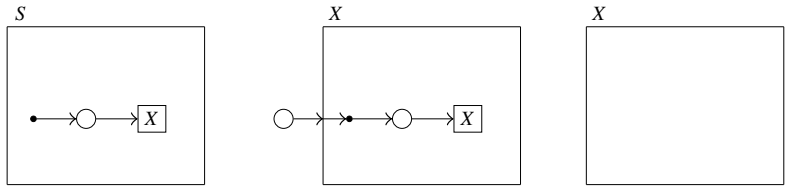
\begin{definition}[Derivation Step \cite{graph_grammar_handbook}]
  For a graph grammar $G=(\Sigma,\Delta,\Gamma,\Omega,P,S)$ and graphs (with
  embedding) $H, H' \in GRE_{\Sigma,\Gamma}$ consider a vertex $v \in V_H$
  and an isomorphic copy $p: X\to(D,C)$ of some production $p' \in P$. We say
  $H \Rightarrow_{v,p} H'$ is a \textit{derivation step} if $\lambda_H(v) = X$ 
  and $H'=H[v/(D,C)]$. If $v$ and $p$ are clear from the context, then we
  write $H \Rightarrow H'$. A sequence of derivation steps 
  $H_0 \Rightarrow_{v_1,p_1} H_2 \Rightarrow_{v_2,p_2} \cdots
  \Rightarrow_{v_n,p_n} H_n$ is called a \textit{derivation}.
  A derivation is called \textit{creative} if $H_0$ and $rhs(p_i)$ are all 
  mutually disjoint. We write $H \Rightarrow_* H'$ if there exists a creative
  derivation from $H$ to $H'$.
\end{definition}

\begin{example}
  A derivation in the above grammar of the string graph with three node vertices:
  \cstikz{tikz/line_derive.tikz}
  where again we color the newly established edges in red.
\end{example}

\begin{definition}[Starting Graph \cite{graph_grammar_handbook}]
  We say that $sn(S,z) \in GR_{\Sigma,\Gamma}$ is a \textit{starting graph} if
  it has only one node given by $z$, its label is $S$, the graph has no edges 
  and no connection instructions.
\end{definition}

\begin{definition}[Graph Grammar Language \cite{graph_grammar_handbook}]
  The graph language induced by a graph grammar
  $G=(\Sigma,\Delta,\Gamma,\Omega,P,S)$ is given by
  $L(G) = \{[H]\ |\ H \in GR_{\Delta,\Omega} \text{ and } sn(S,z)
  \Rightarrow_* H \text{ for some } z\}$, where $[H]$ denotes the equivalence
  class of all graphs which are isomorphic to $H$.
\end{definition}

\begin{definition}[Confluence \cite{graph_grammar_handbook}]
  We say that a graph grammar $G=(\Sigma,\Delta,\Gamma,\Omega,P,S)$ is
  \textit{confluent} if the following holds for every graph $H \in GRE_{\Sigma,
  \Gamma}$:
  if $H \Rightarrow_{u_1,p_1} H_1 \Rightarrow_{u_2,p_2} H_{12}$ and
  $H \Rightarrow_{u_2,p_2} H_2 \Rightarrow_{u_1,p_1} H_{21}$ are creative
  derivations of $G$ with $u_1 \not = u_2$, then $H_{12} = H_{21}$.
\end{definition}

There are simple and easily decidable conditions for determining if an edNCE
grammar is confluent. The class of confluent edNCE grammars is denoted
C-edNCE. % or simply VR (vertex replacement).

\begin{definition}[LIN-edNCE grammar \cite{graph_grammar_handbook}]
  An edNCE grammar $G$ is \textit{linear}, or a LIN-edNCE grammar, if for
  every production $X \to (D,C)$, $D$ has at most one nonterminal node.
\end{definition}

At any given point, only one non-terminal can be replaced, so LIN-edNCE grammars are automatically confluent. Hence, they are a subclass of C-edNCE grammars.

  \presec
  \section{!-graphs vs C-edNCE grammars}\label{sec:encoding}
  \postsec
  % -*- root: main.tex -*-

In this section we outline the relationship between !-graph languages and
context-free languages described by C-edNCE grammars and some of their
respective subclasses. In subsection \ref{sec:nested}, we show that the
language induced by any !-graph $H$ with no overlapping !-boxes can be
described by a LIN-edNCE grammar, which can moreover be constructed effectively
from $H$. In subsection \ref{sec:overlap}, we first show that there exists a
!-graph with trivial overlap of !-boxes whose language cannot be directly
represented using any C-edNCE grammar.  Next, we show that the language of any
!-graph $H$ with trivial overlap between !-boxes can be represented by a
LIN-edNCE grammar up to wire-encoding and the
grammar can be effectively constructed from $H$ as well. In subsection
\ref{sec:unrestricted} we show that the classes of languages induced by
(unrestricted) !-graphs and C-edNCE grammars respectively are incomparable.
% Finally, we conjecture that the only context-free string graph languages
% induced by !-graphs (up to wire-encoding) are
%those given by !-graphs with trivial overlap of !-boxes.
\presec
\subsection{!-graphs without overlap}\label{sec:nested}
\postsec
The main result of this subsection is a theorem stating that the language
induced by any !-graph with no overlapping of !-boxes can be directly
represented by a LIN-edNCE grammar, which can moreover be generated
effectively.

Before we present the main results in this subsection and the next one, we
prove a series of lemmas which are used in the proof of the two main theorems.
Each lemma describes how to build bigger LIN-edNCE grammars out of smaller ones
which in turn describe the language induced by certain subgraphs of a given
!-graph. All of our constructions are effective in the sense that they can be
performed by a computer.

We introduce some conventions that are used throughout our proofs. Given a
!-graph $H$, its vertices will be $\{v_1,v_2,...,v_n\}$. To these vertices, we
will associate non-final edge labels $\{\alpha_1,\alpha_2,...,\alpha_n\}$ which
will be used in the productions of our grammars. Informally, an edge labeled
with $\alpha_i$ will have as its source or target either the original vertex
$v_i$ or one of its copies. This is used in some productions of our grammars,
so that we can easily refer to all copies of such a vertex at once and connect
them to other vertices.

Throughout the proof, we will use grammars satisfying the following conditions:
\begin{enumerate}
  \item The grammar is in LIN-edNCE, thus every production has at most one
    non-terminal node in it
  \item There is a single final production which is the empty graph
  \item Every sentential form with a non-terminal node is such that all
    terminal vertices are connected with an edge to the non-terminal node (in
    both directions) and every such edge has as label some $\alpha_i$
  \item Every production except the final one has connection instructions
    which specify that both incoming and outgoing edges of type $\alpha_i$ are
    connected to the non-terminal. Graphically, we will depict that using
    bidirectional edges as a shorthand notation for an edge in each
    direction.
  \item For every $i$, there is at most one terminal vertex in the productions
    of a grammar which is incident to an edge with label $\alpha_i$
\end{enumerate}
For convenience, we will refer to a grammar satisfying 1-5 as being in \textit{!-linear form}.

\begin{lemma}
  Given a concrete string graph $H$, there exists a LIN-edNCE grammar
  $\mathcal{G}$ which generates the language $\{H\}$. Moreover, this grammar
  can be effectively generated and is in !-linear form.
  \label{lem:base}
\end{lemma}
\begin{proof}
  This is done by the following grammar:
  \cstikz[0.7]{super_simple_grammar.tikz}
  where the set of vertices of $H$ is $V_H = \{v_1, v_2,..., v_n\}$
  and each edge with label $\alpha_i$ has as source or target $v_i$ and the
  non-terminal $F$.
\end{proof}

\begin{lemma}
  Given a !-graph $H$ and a !-linear form grammar $\mathcal{G}$ which generates the same
  language as $H$, there exists a
  grammar $\mathcal{G'}$ which generates the same language as the
  following !-graph:
  \cstikz[0.7]{simple.tikz}
  Moreover, $\mathcal{G'}$ can be effectively generated and is in !-linear form.
  \label{lem:nested}
\end{lemma}
\begin{proof}
  Let's assume that $S$ is the starting production of $\mathcal{G}$, $F$ is the
  final production of $\mathcal{G}$ and that $S'$ and $F'$ are not productions
  of $\mathcal{G}$.

  First, we modify the production $F$ to be the following:
  \cstikz[0.7]{nested_modify.tikz}
  Finally, to the productions of $\mathcal{G}$ we add the following
  productions, where $S'$ is the starting one:
  \cstikz[0.7]{nested_grammar.tikz}
  A derivation
  $S'\Rightarrow S \Rightarrow \cdots \Rightarrow F$
  creates a concrete string graph from the language of $H$, so it is
  simulating a single EXPAND operation applied to the top-level !-box,
  together with a concrete instantiation of the !-boxes in $H$. By
  construction, we can iterate this, thus allowing us to generate multiple
  disjoint concrete graphs, all of which are in the language of $H$.
  A derivation $S' \Rightarrow F'$ simulates the final
  KILL operation applied to the top-level !-box.
\end{proof}

\begin{lemma}
  Given disjoint !-graphs $H,K$ and grammars $\mathcal{G}_1, \mathcal{G}_2$
  which generate the same languages as $H$ and $K$ respectively, then there
  exists a grammar $\mathcal{G}'$ which generates the same language as the
  following !-graph:
  \cstikz[0.7]{disjoint.tikz}
  Moreover, $\mathcal{G'}$ can be effectively generated and is in !-linear form.
  \label{lem:disjoint}
\end{lemma}

\begin{proof}
  Let the vertices of $H$ be $\{v_1,v_2,\ldots,v_k\}$ and let the vertices of $K$
  be $\{v_{k+1},v_{k+2},\ldots,v_n\}$. Also, let $S_i$ and $F_i$ be the starting
  and final productions respectively of $\mathcal{G}_i$. First, we modify each
  production $X$ of $\mathcal{G}_1$, except $F_1$, by adding connection
  instructions for edge labels $\alpha_{k+1},\ldots, \alpha_n$ in the following
  way (left):
  \cstikz[0.7]{disjoint_modify1.tikz}
  where the new additions are colored in red. This doesn't change the language
  of $\mathcal{G}_1$ and is done so that we can put the grammar in the
  required form.
  Similiarly, modify all productions of $\mathcal{G}_2$, except $F_2$, by
  adding to their connection instructions the missing edge labels $\alpha_1,
  \alpha_2, \ldots, \alpha_k$. Finally, modify $F_1$ to be the production
  depicted above (right), so that we can chain together the two grammars.

  The required grammar $\mathcal{G}'$ has as its productions the modified
  productions of $\mathcal{G}_1$ and $\mathcal{G}_2$ with starting production
  $S_1$. A derivation $S_1 \Rightarrow \cdots \Rightarrow F_1$ creates a
  concrete graph from the language of $H$ and a derivation $S_2 \Rightarrow
  \cdots \Rightarrow F_2$ creates a graph from the language of $K$. By
  chaining the two grammars, we simply generate two disjoint concrete graphs,
  one from the language of $H$ and one from the language of $K$, as required.
\end{proof}

\begin{lemma}
  Given !-graph $H$, where $H$ contains a !-box $b$ and given a !-linear form grammar
  $\mathcal{G}$ which generates the same languages as $H$, there exist grammars $\mathcal{G}'$ and $\mathcal{G}''$ which generates
  the same languages as:
  \[ \stikz[0.7]{connected.tikz} \qquad \textrm{and} \qquad
  % and also a grammar $\mathcal{G}''$ which generates the same language as the
  % following !-graph:
  \stikz[0.7]{connected2.tikz} \]
  respectively, where in both cases, the newly depicted edge (colored in red) is incident to the !-box $b$ in $H$ and the edge is also incident to a node-vertex in $H$ which is not in any !-boxes. Moreover, these grammars can be effectively generated and are in !-linear form.
  \label{lem:connect}
\end{lemma}
\begin{proof}
  In both cases, for the newly depicted edge, identify the wire-vertex as $v_i$
  and the node-vertex as $v_j$. To get $\mathcal{G}'$, identify the unique production $X$ of
  $\mathcal{G}$, such that $X$ contains a node-vertex incident to an edge with
  non-final label $\alpha_j$. Then add to its connection instructions a new edge in the following way (left):
  \[
  \stikz[0.7]{connect_modify.tikz}
  \qquad\qquad\qquad
  \stikz[0.7]{connect_modify2.tikz}
  \]
  where the red-colored edge is the new addition. To get the grammar
  $\mathcal{G}''$, follow a similar same procedure (shown on the right above). This modification has the effect that we connect all copies of the wire vertex to the single node vertex (in the appropriate direction), which is the only change required compared to the concrete graphs of $H$.
\end{proof}

\begin{theorem}
  Given a !-graph $H$ such that it doesn't have any overlapping !-boxes,
  there exists a LIN-edNCE grammar $G$ which generates the same language as
  $H$. Moreover, this grammar can be effectively generated and is in !-linear form.
  \label{thm:non-overlap}
\end{theorem}

\begin{proof}
  We present a proof by induction on the number of !-boxes of $H$.

  For the base case, if $H$ has no !-boxes, then lemma \ref{lem:base} completes
  the proof.

  For the step case, pick any top-level !-box and let's consider the full
  subgraph of $H$ it induces. Call this subgraph $K$. Any vertex $v$ of $H
  \backslash K$ which is adjacent to $K$ must be a node-vertex, because
  otherwise this would violate the openness condition of !-boxes.
  Let $w \in K$ be a wire-vertex that $v$ is adjacent to and let $e$ be the
  edge connecting $v$ to $w$. 
  \cstikz[0.7]{thm_no_overlap.tikz}
  If $v$ is in some !-box $b$, then the openness condition of !-boxes
  implies that $w$ must also be in $b$. However, we have assumed that
  $H$ does not contain overlapping !-boxes, so this is not possible and thus
  $v$ is not in any !-boxes. Therefore, we can use lemma \ref{lem:connect} to
  reduce the problem to showing that we can effectively generate a grammar for
  $H \backslash e$. Similarly, by applying the same lemma multiple times, we
  can reduce the problem to showing that we can effectively generate a grammar
  for the !-graph consisting of the disjoint !-graphs $K$ and $H \backslash K$.
  Applying lemma \ref{lem:disjoint} and the induction hypothesis then reduces
  the problem to showing the theorem for $K$.  Finally, we can apply lemma
  \ref{lem:nested} to $K$ and then the induction hypothesis to complete the
  proof.
\end{proof}

\subsection{!-graphs with trivial overlap}\label{sec:overlap}
\postsec

  \begin{wrapfigure}{r}{0.2\textwidth}
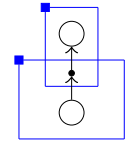

  \centering
  \stikz[0.7]{counter_bang.tikz}
  \caption{!-graph with no direct encoding in
  \textbf{C-edNCE}}\label{fig:bang_k_m_n}
  \end{wrapfigure}
We begin by providing an example of a !-graph with trivial overlap of !-boxes
which cannot be directly represented using any C-edNCE grammar. Then, we prove
the second main theorem of this work which shows that the language of any
!-graph with trivial overlap of !-boxes can be represented by a LIN-edNCE
grammar up to wire-encoding. The construction
of the grammar is also effective.

\begin{proposition}\label{prop:generative}
  The generative power of C-edNCE grammars and Hyperedge Replacement grammars
  on string graphs is the same.
\end{proposition}

\begin{proof}
  The graph $K_{3,3}$ is not a subgraph of any string graph. Then, the
  proposition follows immediately from the main result in \cite{hrl_equal_vr}.
\end{proof}

\begin{proposition}
  The language induced by the !-graph with trivial overlap of !-boxes shown
  in Fig.~\ref{fig:bang_k_m_n} cannot be directly described using any C-edNCE
  grammar.
\end{proposition}

\begin{proof}
  The language of the above !-graph is given by:
  \begin{align*}
  L := \left\{ \stikz[0.65]{k_m_n.tikz} \right\}
  \end{align*}
  In other words, it's the set of complete bipartite graphs $K_{m,n}$ where
  each node-vertex is connected to other node vertices through a single wire
  vertex.

  Let us assume for contradiction that there exists a C-edNCE grammar which
  generates $L$. From proposition \ref{prop:generative}, it follows that $L$
  can be described using a hyperedge replacement grammar. Then, a simple
  application of the pumping lemma for hyperedge-replacement languages
  \cite{pumping} yields a contradiction.
\end{proof}

Note, however, that the graph language of complete bipartite graphs $K_{m,n}$
(without wire vertices in-between) can be generated by a LIN-edNCE grammar.
Moreover, this language is wire-homeomorphic to the language in the above
proposition. However, removing the wire vertices from our C-edNCE language brings another
complication -- we need to account for parallel edges between node vertices by
introducing additional labels on the parallel edges, because in most
of the literature, C-edNCE grammars do not allow for parallel edges of the
same type.

Using these two ideas, we define the notion of wire-encoding and then show
the main result of this subsection.

\begin{definition}[Wire-encoding]
We say that two graphs $H$ and $H'$ are equal up to \textit{wire-encoding},
if we can get one from the other by replacing every edge with special label
$\beta_k$ by a closed wire with endpoints the source and target of the original
edge.
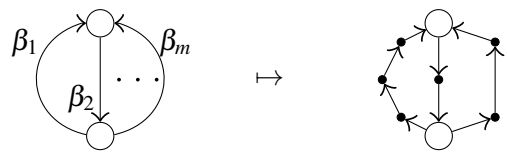
\cstikz{wire_encoding.tikz}
We also say that two graph languages $L$ and
$L'$ are equal up to wire-encoding, if there exists a bijection $f: L \to L'$,
s.t. for every $H \in L$, $H$ and $f(H)$ are equal up to wire-encoding.
\end{definition}

\begin{lemma}
  Given !-graph $H$ which contains non-nested !-boxes $b_1$ and $b_2$ and given a !-linear form
  grammar $\mathcal{G}$ which generates the same languages as $H$, there exists a grammar $\mathcal{G}'$ which
  generates the same language, up to wire-encoding,
  as the following !-graph:
  \cstikz[0.7]{overlap.tikz}
  where the new additions are coloured in red and the newly depicted wire-vertex
  is in both $b_1$ and $b_2$. Moreover, this grammar can be effectively
  generated and is in !-linear form.
  \label{lem:overlap}
\end{lemma}
\begin{proof}
  For the newly depicted edges, identify the source node-vertex as $v_i$ and
  the target node-vertex as $v_j$.  To get the desired grammar $\mathcal{G}'$,
  identify the unique production $X$ of $\mathcal{G}$, such that $X$
  contains a node incident to an edge with non-final label $\alpha_j$. Then,
  add to its connection instructions a new edge with (final) label $\beta_k$:
  \cstikz[0.7]{overlap_final.tikz}
  where the new addition is coloured in red. Note, that with this
  construction, we are not creating the depicted wire-vertex, nor any of its
  copies. We are connecting all copies of the node-vertex $v_i$ to all copies
  of the node-vertex $v_j$ directly with edges labelled with $\beta_p$.

%  Informally, an edge connecting two node vertices in different !-boxes (via a
%  wire-vertex) has the effect that all copies of the source node are connected
%  to all other copies of the target node (via wire vertices).  However, we
%  cannot directly simulate this with our construction of the grammar (see
%  proposition ???). However, we can generate this same language up to
%  wire-homeomorhpism, if we drop the wire-vertex which is in the overlap of the
%  !-boxes and connect the node vertices via edges directly. This brings a
%  further complication, because the two node vertices might be connected via
%  additional edges (via wire vertices) and thus we have to intrdouce extra edge
%  labels to our grammar, because C-edNCE grammars do not allow for parallel
%  edges of the same type.
\end{proof}

\begin{theorem}
  Given a !-graph $H$ such that the only overlap between !-boxes in $H$ is
  trivial, then there exists a LIN-edNCE grammar $G$ which generates the same
  language as $H$, up to wire-encoding.
  Moreover, this grammar can be effectively generated and is in !-linear form.
  \label{thm:overlap}
\end{theorem}

\begin{proof}
  The proof is the same as for the previous theorem, except that we have to
  consider an additional case, namely, when the node-vertex $v$ is in a !-box.
  \cstikz[0.7]{thm_overlap.tikz}
  In this case, the wire-vertex $w$ is in both !-boxes and we can apply lemma
  \ref{lem:overlap} to reduce the problem to showing that $H \backslash w$ can
  be handled. If we set $W$ to be the set of all wire vertices which overlap
  with $b$ and another !-box, then applying the same lemma multiple times
  reduces the problem to showing that $H \backslash W$ can be handled. However,
  we have assumed that only trivial overlap between !-boxes exists in $H$ and
  therefore in $H \backslash W$ there is no overlap between $b$ and any other
  !-boxes. Then, the proof can be finished using the same arguments as in
  the previous theorem.
\end{proof}

\subsection{The power of context-free languages}\label{sec:unrestricted}
\postsec

In this subsection, we will show that there exists a LIN-edNCE language
which cannot be induced by any (unrestricted) !-graph.
We start by providing some definitions.

\begin{definition}[Maximum distance]
  For a graph $G$ and vertices $v,u \in G$, the \textit{distance} between $u$
  and $v$ is the length of the shortest path connecting $u$ and $v$. If there
  is no path between $u$ and $v$ then we say that the distance is -1. The
  distance between a vertex and itself is 0. The \textit{maximum distance} for
  a graph $G$ is the largest distance among all pairs of vertices.
\end{definition}

\begin{definition}[Bounded maximum distance]
  For a set of graphs $\mathcal{G} = \{G_1, G_2, G_3,...\}$, we say that
  $\mathcal{G}$ is of \textit{bounded maximum distance} if there exists an
  integer $n \in \mathbb{N}$, such that the maximum distance for every graph in
  $\mathcal{G}$ is smaller than $n$. If such an $n$ does not exist, then we say
  that $\mathcal{G}$ is of \textit{unbounded maximum distance}.
\end{definition}

\begin{proposition}
  The language induced by any !-graph is of bounded maximum distance.
\end{proposition}

\begin{proof}
  Consider a !-box $B \subset G$.
  Applying a KILL operation to $B$ cannot increase the maximum
  distance. Applying an EXPAND operation once could potentially increase it,
  however, applying an EXPAND more than once will not increase it any further.
  Thus, regardless of how many times an EXPAND operation is applied to $B$
  the maximum distance can only increase by a fixed amount. Also, because of
  the symmetric properties of the EXPAND map, any nested !-boxes within $B$
  or any overlapping !-boxes can increase the maximum distance with a fixed
  amount as well regardless of how many EXPAND operations are applied to any of
  them.

  By combining the above observation with the fact that $G$ has finitely many
  !-boxes we can conclude that the set of graphs induced by $G$ is of bounded
  maximum distance.
\end{proof}

In the next proposition, we show that even severely restricted graph grammars
can generate languages which are not induced by any !-graph, thereby
establishing that \textbf{C-edNCE}\ $\not \subseteq$\ \textbf{BG}.

\begin{proposition}
  The language of the LIN-edNCE grammar of Example~\ref{ex:line_graph}:
  \cstikz[0.7]{unbounded.tikz}
  is not induced by any !-graph, even up to wire-encoding.
\end{proposition}

\begin{proof}
  The language generated by this grammar is of unbounded maximum distance and
  thus it cannot be generated by any !-graph.
\end{proof}

% \begin{remark}
%   The above grammar is actually a linear \textit{apex} grammar, in the class LIN$\cdot$A-edNCE which are special kinds
%   of LIN-edNCE grammars.
% \end{remark}

  \presec
  \section{Conclusion and future work}\label{sec:fw}
  \postsec
  % -*- root: main.tex -*-

% \begin{figure}
% \cstikz{venn_diagram.tikz}
% \caption*{Venn diagram of Lin-edNCE, Bang Graph and Trivial Overlap Bang Graph
% languages}
% \end{figure}

\begin{wrapfigure}{R}{0.4\textwidth}
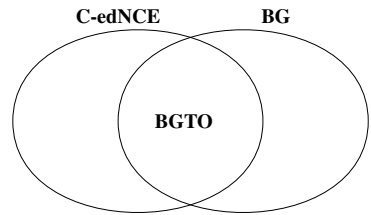

\cstikz[0.7]{venn_diagram2.tikz}
\caption{Conjectured relationship between \textbf{C-edNCE} and
  \textbf{BG}}\label{fig:venn2}
\end{wrapfigure}
We have shown that the language of a !-graph with trivial overlap can always be encoded using a context-free grammar. \textbf{BGTO} languages are in fact a natural restriction to \textbf{BG}, and have already arisen in alternative !-box formalisations, such as the (non-commutative) !-tensors described in~\cite{Kissinger:2014ab}. Thus, gaining a better understanding of \textbf{BGTO} languages is valuable in its own right. The presence of non-trivial overlap tends to cause highly non-local effects when expanding !-boxes, so it could be the case that the property of being context-free actually \textit{characterises} trivial overlap. In other words:

\begin{conjecture}
  The language induced by any !-graph which contains !-boxes whose overlap is non-trivial cannot be described by a C-edNCE grammar, even up to wire-encoding.
\end{conjecture}

If the conjecture holds, then the Venn Diagram in Section~\ref{sec:intro} would simplify to the one in fig.~\ref{fig:venn2},
which lends credence to the notion that BGTO languages are the string graph analogue to regular languages.

Quite aside from classification issues, the next step for context-free string graphs grammars is to develop the tools for working with them. For example, since they are built on top of string graphs, !-graphs can be plugged together along inputs and outputs to get new !-graphs, and more importantly, they can be used to define \textit{!-graph rewrite rules}:
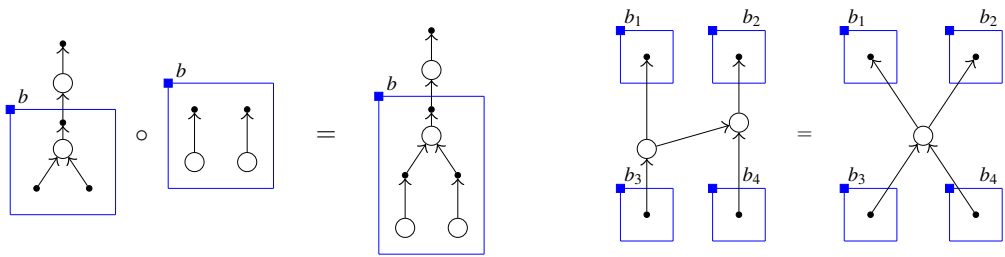
\[ \stikz[0.7]{bang_graph-plug.tikz} \qquad\qquad
\stikz[0.7]{spider-merge-bb.tikz} \]
Just as !-graphs represent families of string graphs, !-graph rewrite rules represent families of string graph rewrite rules. Furthermore, new !-graph rules can be introduced from concrete ones using a technique called !-box induction~\cite{merry_dphil}. The natural next step in this program then is to see if composition, rewriting, and a notion of graphical induction carry through to the context-free case.

\noindent \textbf{Acknowledgements.} We would like to thank the anonymous
referees for their feedback and we wish to gratefully acknowledge financial
support from EPSRC, the Scatcherd European Scholarship, and the John Templeton
Foundation.
\presec

%\presec
%\def\baselinestretch{0.85}
{\footnotesize
  \bibliographystyle{eptcs}
  \bibliography{vladimir_refs}{}
}

\end{document}